\documentclass[12pt]{article}

\usepackage{amsfonts}
\usepackage{amsmath}
\usepackage{amscd}
\usepackage{amssymb}
\usepackage{amsthm}
\usepackage{amsxtra}
\usepackage{diagrams}

\newcommand{\bbC}    {\mathbb C}
\newcommand{\bbR}    {\mathbb R}
\newcommand{\bbF}     {\mathbb F}

\newcommand{\bbZ}    {\mathbb Z}
\newcommand{\bbN}    {\mathbb N}
\newcommand{\bbI}      {\mathbb I}

\newcommand{\cC}    {{\cal C}}

\newcommand{\cE}    {{\cal E}}
\newcommand{\cF}    {{\cal F}}

\newcommand{\cS}    {{\cal S}}

\newcommand{\al}    {\alpha}
\newcommand{\be}    {\beta}
\newcommand{\de}    {\delta}

\newcommand{\om}    {\omega}
\newcommand{\Ga}    {\Gamma}
\newcommand{\La}    {\Lambda}
\newcommand{\Om}    {\Omega}

\newcommand{\vte}    {\vartheta}

\newcommand{\A}      {\mathrm A}
\newcommand{\B}       {\mathrm B}
\newcommand{\D}       {\mathrm D}

\newcommand{\M}         {\mathrm M}

\newcommand{\X}       {\mathrm X}

\newcommand{\U}        {\mathrm U}

\newcommand{\bA}     {\mathbf A}
\newcommand{\bB}      {\mathbf B}

\newcommand{\mbf}       {\mathbf f}
\newcommand{\bg}       {\mathbf g}

\newcommand{\bU}     {\mathbf U}

\newcommand{\bu}    {\mathbf u}

\newcommand{\bfI}     {\mathbf I}
\newcommand{\bfJ}     {\mathbf J}
\newcommand{\bj}     {\mathbf j}

\newcommand{\fD}    {\mathfrak D}

\newcommand{\fF}      {\mathfrak F}

\newcommand{\w}       {\wedge}
\newcommand{\ck}      {\check} 
\newcommand{\na}      {\nabla} 
\newcommand{\bLa}    {\pmb{\Lambda}}

\newcommand{\bphi}   {\pmb{\phi}}
\newcommand{\bchi}   {\pmb{\chi}}
\newcommand{\bpsi}   {\pmb{\psi}}
\newcommand{\bxi}     {\pmb{\xi}}
\newcommand{\bet}     {\pmb{\eta}} 
\newcommand{\bvte}   {\pmb{\vartheta}}
\newcommand{\bde}     {\pmb{\delta}}
\newcommand{\p}        {\partial}

\newcommand{\lan}      {\langle}
\newcommand{\ran}      {\rangle}
 
\newcommand{\Int}       {\textstyle{\int\!}}
\newcommand{\bcE}     {\pmb{\cal E}} 

\newcommand{\oF}       {\overset{\,\circ}\cF{}}

\DeclareMathOperator{\Hom}   {Hom}
\DeclareMathOperator{\End}   {End}

\DeclareMathOperator{\ord}   {ord}

\DeclareMathOperator{\im}    {Im}
\DeclareMathOperator{\ke}    {Ker}
\DeclareMathOperator{\com}   {{\scriptstyle\circ}\,}

\DeclareMathOperator{\ev}    {ev}

\DeclareMathOperator{\Div}     {Div}

\newtheorem{theorem}{Theorem}
\newtheorem{lemma}{Lemma}
\newtheorem{prop}{Proposition}
\newtheorem{cor}{Corollary}
\theoremstyle{definition} 
\newtheorem{rem}{Remark} 

\newtheorem{defi}{Definition}
\newtheorem{ass}{Assumption}

\begin{document}

\title{Lie\,-Poisson structures over differential algebras}
\author
          {
             Zharinov V.V.
             \thanks{Steklov Mathematical Institute}
             \thanks{E-mail: zharinov@mi.ras.ru}
            }
\date{}
\maketitle

\begin{abstract}
In this paper we use key elements of the Olver's approach to Hamiltonian evolution 
equations in partial derivatives and propose an algebraic construction appropriate 
for Hamiltonian evolution systems with constraints.
\end{abstract} 


\section{Introduction.} 
In his book \cite{PJO}, Chapter VII,  P.J. Olver has developed a Hamiltonian approach 
to evolution systems of differential equations in partial derivatives. 
In essence, the main component of his technics is the construction of a Lie-Poisson structure 
over the differential algebra of the differential functions (see Subsection \ref{ODA}, below), 
associated with evolution systems without constraints. 

In this paper we use key elements of the Olver's scheme and present a general algebraic 
construction appropriate, in particular, for Hamiltonian evolution systems of partial differential equations with constraints. 
Namely, in a pure algebraic manner we describe Lie-Poisson structures over 
arbitrary differential algebras. The main result is Theorem \ref{MT}, 
the algebraic analog of the Theorem 7.8 from the book \cite{PJO}, 
characterizing Hamiltonian operators over differential algebra under study. 
Our abstract results can be used as follows. 
For a given system of constraints (a system of partial differential equations in the space 
variables) one should  define an appropriate differential algebra and then search for 
Hamiltonian operators compatible with the given evolution system. 
Note, both these problems are rather difficult. 
Moreover, the second one is solvable extremely rare, 
each solvable case is a  great success and an important step in the study 
of the given evolution system with constraints.  

We use the following general notations: 
\begin{itemize} 
	\item 
		$\bbF=\bbR,\bbC$, \quad $\bbN=\{1,2,3,\dots\}\subset\bbZ_+=\{0,1,2,\dots\}$; 
	\item 
		$\M=\{1,\dots,m\}$,  \quad $m\in\bbN$;
	\item 
		$\bbI=\bbZ^\M_+=\{i=(i^1,\dots,i^m)\mid i^\mu\in\bbZ_+, \ \mu\in\M\}$.
\end{itemize} 

\section{Algebraic essentials.} 

\subsection{The basic algebra.} 
Let $\cF$ be an associative commutative unital algebra over the number field  $\bbF$.
We denote by $\fD=\fD(\cF)$ the set of all {\it differentiations} of the algebra $\cF$, 
i.e., all $\bbF$-linear mappings $X :\cF\to\cF$ (i.e., $X\in\End_\bbF(\cF)$), 
satisfying the {\it Leibniz rule} 
\begin{equation*}
	X(f\cdot g)=(Xf)\cdot g+f\cdot(Xg) \quad \text{for all}\quad f,g\in\cF. 
\end{equation*}
The set $\fD$ has two matching structures, namely, 
the structure of a Lie algebra with the commutator $X,Y\mapsto[X,Y]=X\com Y-Y\com X$
as the Lie bracket, 
the structure of a $\cF$-module with $(f\cdot X)g=f\cdot(Xg)$, 
and the matching condition $[X,g\cdot Y]=(Xg)\cdot Y+g\cdot[X,Y]$ 
for all $f,g\in\cF$, $X,Y\in\fD$. 

For any index sets $\A,\B$ the $\cF$-module 
\begin{equation*}
	\cF^\A_\B=\big\{\eta=(\eta^\al_\be) \ \big| \ \eta^\al_\be\in\cF, \ \al\in\A, \ \be\in\B\big\} 
\end{equation*}
is defined, where the multiplication $f\cdot\eta$, $f\in\cF$, $\eta\in\cF^\A_\B$,  
is defined component-wise, 
i.e., $(f\cdot\eta)^\al_\be=f\cdot\eta^\al_\be$, $\al\in\A$, $\be\in\B$. 

For any $\cF$-module $\cF^\A_\B$ and any differentiation  $X\in\fD$ the linear mapping
$X : \cF^\A_\B\to\cF^\A_\B$ is defined component-wise, 
$\eta=(\eta^\al_\be)\mapsto X\eta=(X\eta^\al_\be)$. 
In particular, the Leibniz rule takes the form: 
$X(f\cdot\eta)=(Xf)\cdot\eta+f\cdot(X\eta)$, for all $X\in\fD$, $f\in\cF$, $\eta\in\cF^\A_\B$. 

We denote by $\oF^\A_\B$ the $\cF$-module of all elements 
$\eta=(\eta^\al_\be)\in\cF^\A_\B$, s.t. for any upper index $\al\in\A$ 
only a finite number of components $\eta^\al_\be\ne0$. 
In particular, $\oF^\A=\cF^\A$ for any index set $\A$, 
while $\oF_\B$ consists of all finite elements $\eta=(\eta_\be)\in\cF_\B$, 
and $\oF_\B=\cF_\B$ iff (i.e., if and only if) the index set $\B$ is finite.  

Elements $\phi=(\phi^\al)\in \cF^A$ with the upper indices (superscripts) 
may be considered as ``vectors'', elements $\xi=(\xi_\al)\in\oF_A$ with the lower indices 
(subscripts) may be considered as ``covectors'', and elements with multiple indices 
may be considered as ``tensors''. 
The summation over repeated upper and lower indices is assumed 
and it may be considered as ``pairing'' (contraction). 
The index sets may be infinite, 
but we shall ensure that all formally infinite sums in fact will be finite. 

\subsection{The de Rham complex.}
The $\cF$-modules $\Om^q=\Om^q(\cF)$, $q\in\bbZ_+$, of {\it $q$-forms over 
the algebra $\cF$ with values in $\cF$} are defined as follows: 
\begin{equation*} 
	\Om^0=\cF, \quad \Om^q=\Hom_\cF(\w^q_\cF\fD;\cF), \quad q\in\bbN.
\end{equation*}
The direct sum $\Om=\oplus_{q\in\bbZ_+}\Om^q$ (here $\oplus=\oplus_\cF$)
has the natural structure of an exterior algebra. 
The {\it exterior differential} $d\in\End_\bbF(\Om)$ is defined by the {\it Cartan formula}: 
\begin{align*} 
	d\om(X_0,\dots,X_q)
	=&\frac1{q+1}\bigg\{\sum_{0\le r\le q}(-1)^r X_r\big(\om(X_0,\dots\ck{X_r}\dots,X_q)\big) \\
	+&\sum_{0\le r<s\le q}(-1)^{r+s}\om([X_r,X_s],X_0,\dots\ck{X_r}\dots\ck{X_s}\dots,X_q)\bigg\}
\end{align*}
for all $q\in\bbZ_+$, $\om\in\Om^q$, $X_0,\dots,X_q\in\fD$ 
(here and hereafter, the ``checked'' arguments are understood to be omitted). 

Note, $d^q=d|_{\Om^q} : \Om^q\to\Om^{q+1}$, $q\in\bbZ_+$, and $d\com d=0$. 
Thus, the {\it de Rham complex $\{\Om^q;d^q\}$ of the algebra $\cF$} is defined 
(see, e.g., \cite{Z11}). 
The {\it cohomology spaces} of this complex are $H^q=H^q(\cF)=\ke d^q\big/\im d^{q-1}$, 
where $q\in\bbZ_+$, $\ke d^q=\{\om\in\Om^q\mid d\om=0\}$, 
$\im d^{q-1}=\{\om=d\chi\mid \chi\in\Om^{q-1}\}=d\Om^{q-1}$ ($\Om^{-1}=0$). 

\subsection{The differential algebra. The basic assumptions.}
The algebra $\cF$ is called {\it differential} if a subalgebra $\fD_H=\fD_H(\cF)$ 
of the Lie algebra $\fD$ is specified with a $\cF$-basis $D=\{D_\mu\mid \mu\in\M\}$, 
where commutators $[D_\mu,D_\nu]=0$, $\mu,\nu\in\M$.
Note, a differential algebra $\cF$ is in fact the pair $(\cF,D)$ 
(see, e.g., \cite{Z12},\cite{Z2}). 

Let $(\cF,D)$ be a differential algebra under study. 
\begin{ass} 
We assume that the $\cF$-module $\fD$ is splitted into a direct sum 
$\fD=\fD_V\oplus_\cF\fD_H$ 
(here and hereafter, $V$ stands for the {\it vertical} component 
while $H$ stands for the {\it horizontal} component), 
where $\fD_V=\fD_V(\cF)$ is a subalgebra of $\fD$, 
such that $[D,\fD_V]\subset\fD_V$ 
(i.e. $[D_\mu,V]\in\fD_V$ for all $V\in\fD_V$ and $\mu\in\M$). 
Moreover, we assume that the Lie algebra $\fD_V$  has a formal $\cF$-basis 
$\{\p_a\mid a\in\bA\}$, where $[\p_a,\p_b]=0$, $a,b\in\bA$, 
$\bA$ is an index set.  
In this case, $[D_\mu,\p_a]=\Ga^b_{\mu a}\cdot\p_b$, 
where the symbol $\Ga=(\Ga^b_{\mu a})\in\oF^\bA_{\M\bA}$, $a,b\in\bA$, $\mu\in\M$. 
Note, that here the set $\M$ is finite, while the set $\bA$ is as a rule infinite, 
and according to the above definition of the $\cF$-module $\oF^\bA_{\M\bA}$, 
for any index $b\in\bA$ only a finite number of coefficients $\Ga^b_{\mu a}\ne0$. 
\end{ass}

\begin{ass}
We assume that for any $f\in\cF$ the action $\p_a f\ne0$ 
only for a finite number of indices $a\in\bA$ (i.e. $\p f=(\p_a f)\in\oF_\bA$). 
In particular, $\fD_V=\{V=\phi^a\cdot\p_a\mid \phi^a\in\cF\}$, and the action 
of the formally infinite sum $\phi^a\cdot\p_a$ on $\cF$ is well defined. 
\end{ass} 

\begin{ass} 
We assume that the differential algebra $(\cF,D)$ is of the {\it du Bois-Reymond type}, 
i.e., it has the property: 
the equality $\phi\cdot\psi=D_\mu\chi^\mu$ is valid for a fixed $\phi\in\cF$ and all 
$\psi\in\cF$ with some $\chi^\mu=\chi^\mu(\phi,\psi)\in\cF$, $\mu\in\M$, iff  $\phi=0$. 
\end{ass}

Let us denote by 
$\fD_E=\fD_E(\cF)=\{V\in\fD_V\mid [D_\mu,V]=0, \ \mu\in\M\}$ the set of all {\it evolutionary}
differentiations of the algebra $\cF$. 
The set $\fD_E$ is a subalgebra of the Lie algebra $\fD_V$, because 
$[\fD_E,\fD_E]\subset\fD_E$ due to the {\it Jacobi identity} for commutators, 
but it is not a submodule of the $\cF$-module $\fD_V$ (see, e.g., \cite{Z1}). 

Let us consider the  set $\bcE=\ke\na=\{\bphi=(\phi^a)\in\cF^\bA\mid \na\bphi=0\}$,
where 
\begin{equation*}
	\na : \cF^\bA\to\cF^\bA_\M, \quad \bphi=(\phi^a)\mapsto\bet=(\eta^a_\mu), 
		\quad \eta^a_\mu=D_\mu\phi^a+\Ga^a_{\mu b}\cdot\phi^b.
\end{equation*}
Clear, the set $\bcE$ has the structure of a linear subspace 
of the linear space $\cF^\bA$. 

\begin{prop} 
The mapping $\ev : \bcE\to\fD_E$, $\bphi=(\phi^a)\mapsto U=\ev_{\bphi}=\phi^a\cdot\p_a$, 
is an isomorphism of linear spaces. 
Moreover, the structure of the Lie algebra on $\fD_E$ 
defines the isomorphic structure on $\bcE$ by the rule: 
\begin{equation*}
	W=[U,V]\mapsto\bxi=[\bphi,\bpsi],   
	\quad U=\phi^a\cdot\p_a, \quad V=\psi^a\cdot\p_a, \quad W=\xi^a\cdot\p_a, 
\end{equation*}
where $\xi^a=U\psi^a-V\phi^a$, $a\in\bA$. 
\end{prop} 

\subsection{The differential bicomplex.}
The splitting $\fD=\fD_V\oplus_\cF\fD_H$ of differentiations 
leads to the splitting of the differential forms, 
$\Om=\oplus_{p,q\in\bbZ_+}\Om^{pq}$,  where 
\begin{equation*} 
	\Om^{pq}=\Om^p_V\w_\cF\Om^q_H, \quad 
	\Om^p_V=\Hom_\cF(\w^p_\cF\fD_V;\cF), \quad 
	\Om^q_H=\Hom_\cF(\w^q_\cF\fD_H;\cF), 
\end{equation*} 
in particular, $\Om^{00}=\Om^0_V=\Om^0_H=\cF $ (here, $\oplus=\oplus_\cF$). 

In the same way, the exterior differential $d : \Om\to\Om$ 
splits in the two components: $d=d_V+d_H$, where 
\begin{align*} 
	d^{pq}_V
		&=d_V\big|_{\Om^{pq}} : \Om^{pq}=\Om^p_V\w_\cF\Om^q_H
			\to\Om^{p+1}_V\w_\cF\Om^q_H=\Om^{p+1,q}, \\  
	d^{pq}_H
		&=d_H\big|_{\Om^{pq}} : \Om^{pq}=\Om^p_V\w_\cF\Om^q_H
			\to\Om^p_V\w_\cF\Om^{q+1}=\Om^{p,q+1}. 		
\end{align*}
The identity $d\com d=0$ implies the identities: 
\begin{equation*} 
	d_V\com d_V=d_V\com d_H+d_H\com d_V=d_H\com d_H=0.
\end{equation*}

Further, the $\cF$-module $\Om^1_V=\Hom_\cF(\fD_V;\cF)=\fD_V^*$ 
has the dual $\cF$-basis $\{\rho^a\mid a\in\bA\}$, $\rho^a(\p_b)=\de^a_b$, $a,b\in\bA$, 
so $\Om^1_V=\{\om_a\cdot\rho^a\mid \om=(\om_a)\in\oF_\bA\}$. 
In general, for any $p\in\bbN$ the $\cF$-module 
\begin{equation*}
	\Om^p_V=\w^p_\cF\Om^1_V
	=\big\{\om=\om_{a_1\dots a_p}\cdot\rho^{a_1}\w\dots\w\rho^{a_p} 
	\ \big| \ \om_{a_1\dots a_p}\in\cF\big\}, 
\end{equation*}
where only a finite number of coefficients $\om_{a_1\dots a_p}\ne0$. 

In the same way, the $\cF$-module $\Om^1_H=\Hom_\cF(\fD_H;\cF)=\fD_H^*$ 
has the dual $\cF$-basis $\{\vte^\mu\mid \mu\in\M\}$, $\vte^\mu(D_\nu)=\de^\mu_\nu$, 
so $\Om^1_H=\{\om_\mu\cdot\vte^\mu\mid \om_\mu\in\cF\}$.
For any $q\in\bbN$ we have $\Om^q_H=\w^q_\cF\Om^1_H$ 
(in fact, $\Om^q_H=0$ for $q>m$). 

The equality $d(\om\w\chi)=(d\om)\w\chi+(-1)^{\ord\om}\om\w(d\chi)$, 
$\om,\chi\in\Om$, implies the analogous equalities for $d_V$ and $d_H$. 

It is easy to verify that 
\begin{itemize} 
	\item 
		$d_V L=\p_a L\cdot\rho^a$, \quad $d_H L=D_\mu L\cdot \vte^\mu$, 
		\quad $ L\in\Om^0=\cF$; 
	\item 
		$d_V\rho^a=0$, \quad 
		$d_H\rho^a=\Ga^a_{\mu b}\cdot \rho^b\w\vte^\mu$, 
		\quad $a\in\bA$, \ $\mu\in\M$; 
	\item 
		$d_V\vte^\mu=0$, \quad $d_H\vte^\mu=0$, \quad $\mu\in\M$.
\end{itemize}
Hence, for example, let 
$\om=\om_a\cdot\rho^a\in\Om^1_V$, $\chi=\chi_\mu\cdot\vte^\mu\in\Om^1_H$, 
then 
\begin{itemize} 
	\item 
		$d_V\om=d_V\om_a\w\rho^a+\om_a\cdot d_V\rho^a
		=\p_{[a}\om_{b]}\cdot\rho^a\w\rho^b$, \\
		where $\p_{[a}\om_{b]}=\frac12\big(\p_a\om_b-\p_b\om_a\big)$; 
	\item 
		$d_H\om=d_H\om_a\w\rho^a+\om_a\cdot d_H\rho^a
		=(-D_\mu\om_a+\Ga^b_{\mu a}\cdot\om_b)\cdot\rho^a\w\vte^\mu$; 
	\item 
		$d_V\chi=d_V\chi_\mu\w\vte^\mu+\chi_\mu\cdot d_V\vte^\mu
		=\p_a\chi_\mu\cdot\rho^a\w\vte^\mu$ ; 
	\item 
		$d_H\chi=d_H\chi_\mu\w\vte^\mu+\chi_\mu\cdot d_H\vte^\mu
		=\p_{[\mu}\chi_{\nu]}\cdot\vte^\mu\w\vte^\nu$, \\
		where $\p_{[\mu}\chi_{\nu]}=\frac12\big(\p_\mu\chi_\nu-\p_\nu\chi_\mu\big)$.
\end{itemize} 

In particular, the bicomplex $\{\Om^{pq};d^{pq}_V,d^{pq}_H\}$ is defined 
with the cohomology spaces 
$H^{pq}_V=\ke d^{pq}_V\big/\im d^{p-1,q}_V$ and 
$H^{pq}_H=\ke d^{pq}_H\big/\im d^{p,q-1}_H$, $p,q\in\bbZ_+$. 

\subsection{The basic ingredients.}
We need the bottom-right corner of the above bicomplex: 
\begin{diagram}
        &                               &\vdots                       &                               &\vdots                   &     &  \\
        &                               &\uTo^{d^{1,m-1}_V}&                                &\uTo^{d^{1m}_V}&     &  \\                         
\dots&\rTo^{d^{1,m-2}_H}&\Om^{1,m-1}           &\rTo^{d^{1,m-1}_H}&\Om^{1m}           &\rTo&0\\
        &                               &\uTo^{d^{0,m-1}_V}&                                &\uTo^{d^{0m}_V}&     &  \\
\dots&\rTo^{d^{0,m-2}_H}&\Om^{0,m-1}           &\rTo^{d^{0,m-1}_H}&\Om^{0m}           &\rTo&0\\
        &                               &\uTo                         &                                &\uTo                     &      &  \\
        &                              &0                               &                                &0                          &     &
\end{diagram}
Let us set 
$\bvte=\vte^1\w\dots\w\vte^m$, 
$\bvte_\mu=(-1)^{\mu-1}\vte^1\w\dots\ck{\vte^\mu}\dots\w\vte^m$, 
so $\vte^\mu\w\bvte_\nu=\de^\mu_\nu\bvte$, $\mu,\nu\in\M$. 
With these notations we have: 
\begin{itemize} 
	\item 
		$\Om^{0,m-1}=\{\om=\psi^\mu\cdot\bvte_\mu\mid \psi^\mu\in\cF\}\simeq\cF^\M$; 
	\item 
		$\Om^{0m}=\{\om=L\cdot\bvte\mid L\in\cF\}\simeq\cF$; 
	\item 
		$\Om^{1,m-1}
		=\big\{\om=\chi^\mu_a\cdot\rho^a\w\bvte_\mu \ \big| \  
		\bchi=(\chi^\mu_a)\in\oF^\M_\bA\big\}\simeq\oF^\M_\bA$; 
	\item 
		$\Om^{1m}=\big\{\om=f_a\cdot\rho^a\w\bvte \ \big| \ 
		\mbf=(f_a)\in\oF_\bA\big\}\simeq\oF_\bA$. 
\end{itemize}
Further, 
\begin{itemize} 
	\item 
		$d^{0,m-1}_H : \Om^{0,m-1}\to\Om^{0m}$, 
		$\om=\psi^\mu\cdot\bvte_\mu\mapsto d_H\om=\Div\psi\cdot\bvte$; 
	\item 
		$d^{1,m-1}_H : \Om^{1,m-1}\to\Om^{1m}$, 
		$\om=\chi^\mu_a\cdot\rho^a\w\bvte_\mu\mapsto d_H\om=(\na^*\bchi)_a\cdot\rho^a\w\bvte$; 
	\item 
		$d^{0m}_V : \Om^{0m}\to\Om^{1m}$, 
		$\om=L\cdot \bvte\mapsto d_V\om=(\p L)_a\cdot\rho^a\w\bvte$; 		
\end{itemize}
where 
\begin{itemize} 
	\item 
		$\Div=-D^* : \cF^\M\to\cF$, \quad $\psi=(\psi^\mu)\mapsto\Div\psi=D_\mu\psi^\mu$ \\ 
		\big($D : \cF\to\cF_\M$, \quad $L\mapsto\zeta=(\zeta_\mu), \quad \zeta_\mu=D_\mu L$\big); 
	\item 
		$\na^* : \oF^\M_\bA\to\oF_\bA$, \quad $\bchi=(\chi^\mu_a)\mapsto\na^*\bchi$, 
		\quad $(\na^*\bchi)_a=-D_\mu\chi^\mu_a+\Ga^b_{\mu a}\cdot\chi^\mu_b$; 
	\item 
		$\p :  \cF\to\oF_\bA$, \quad  $L\mapsto\p L$, \quad $(\p L)_a=\p_a L$.
\end{itemize} 
Thus, 
\begin{itemize} 
	\item 
		$H^{0m}_H=\Om^{0m}\big/d^{0,m-1}_H\Om^{0,m-1}=\cF\big/\Div\cF^\M$, 
		\quad $\Div\cF^\M=\im\Div$; 
	\item 
		$H^{1m}_H=\Om^{1m}\big/d^{1,m-1}_H\Om^{1,m-1}
		=\oF_\bA\big/\na^*\oF^\M_\bA$,
		\quad $\na^*\oF^\M_\bA=\im\na^*$; 
	\item 
		$[d^{0m}_V] : H^{0m}_H\to H^{1m}_H$, 	\quad  
		$[\om]\mapsto[d_V][\om]=[d_V\om]$, \\
		i.e. $[d^{0m}_V] : \cF\big/\Div\cF^\M\to\oF_\bA\big/\na^*\oF^\M_\bA$, \quad 
		$[L]\mapsto[\p L]$.
\end{itemize} 
We shall use the notation 
\begin{itemize} 
	\item 
		$\fF=H^{0m}_H=\cF\big/\Div\cF^\M
		       =\big\{\Int L=[L]=L+\Div\cF^\M \ \big| \ L\in\cF\big\}$; 
	\item 
		$\cS=H^{1m}_H=\oF_\bA\big/\na^*\oF^\M_\bA
		       =\big\{[\mbf]=\mbf+\na^*\oF^\M_\bA \ \big| \ \mbf\in\oF_\bA \}$; 
	\item 
		$\bde=[d^{0m}_V] : \fF\to\cS$, \quad $\Int L\mapsto\bde\Int L
		=[\p L]=\p L+\na^*\oF^\M_\bA$. 
\end{itemize}
There are the natural pairings 
\begin{itemize} 
	\item 
		$\lan\cdot,\cdot\ran : \oF_\bA\times\cF^\bA\to\cF$, \quad
		$\mbf=(f_a), \ \bphi=(\phi^a)\mapsto\lan\mbf,\bphi\ran=f_a\cdot\phi^a$; 
	\item 
		$\lan\cdot,\cdot\ran : \oF^\M_\bA\times\cF^\bA_\M\to\cF$, \quad 
		$\bchi=(\chi^\mu_a), \ \bet=(\eta^a_\mu)\mapsto\lan\bchi,\bet\ran=\chi^\mu_a\cdot\eta^a_\mu$;  
\end{itemize} 
and the linear mapping (the natural projection in the exact sequence of the  $\bbF$-linear spaces 
$0\to\Div\cF^\M\to\cF\to\fF\to0$) 
\begin{itemize} 
	\item 
		$\int : \cF\to\fF, \quad L\mapsto\Int L=L+\Div\cF^\M$.
\end{itemize} 

It is to verify the following two statements. 

\begin{prop} 
For any $\bphi\in\bcE$ the image $\ev_{\bphi}\Div\cF^\M\subset\Div\cF^\M$, 
in particular the mapping $\ev_{\bphi} : \fF\to\fF$ is defined by the rule: 
$\ev_{\bphi}\Int K=\Int\ev_{\bphi} K$ for any $K\in\cF$.
\end{prop}

\begin{prop} 
The mapping $\na^* : \oF^\M_\bA\to\oF_\bA$ is the {\it Lagrange dual} 
for the mapping $\na : \cF^\bA\to\cF^\bA_\M$, 
i.e., for any $\bchi\in\oF^\M_\bA$ and $\bphi\in\cF^\bA$ 
the {\it Green formula} $\lan\bchi, \na\bphi\ran-\lan\na^*\bchi,\bphi\ran=\Div\psi$ 
holds, where $\psi=(\psi^\mu)\in\cF^\M$, $\psi^\mu=\chi^\mu_a\cdot\phi^a$ 
(in other words, $\Int(\lan\bchi, \na\bphi\ran-\lan\na^*\bchi,\bphi\ran)=0$).
\end{prop}  

We shall need the dual $\bbF$-linear space 
\begin{align*} 
	\cS^*=\Hom_\bbF(\cS;\fF)
		  &=\big\{\bphi\in\cF^\bA \ \big| \ \Int\lan\na^*\oF^\M_\bA,\bphi\ran=0\big\} \\
		  &=\big\{\bphi\in\cF^\bA \ \big| \ \Int\lan\oF^\M_\bA,\na\bphi\ran=0\big\}=\bcE, 
\end{align*}
the last equality due to the du Bois-Reymond property of the differential algebra $(\cF,D)$ 
(remind, $\bcE=\ke\na$). 
Note, that for any  $\bphi\in\cF^\bA$ the $\bbF$-linear mapping $\bphi : \cS\to\fF$ 
is defined by the rule: 
$[\mbf]\mapsto\Int\lan\mbf,\bphi\ran$ for all $\mbf\in\oF_\bA$. 
In particular, there is defined the pairing 
\begin{equation*} 
	\cS\times\bcE\to\fF, \quad [\mbf],\bphi\mapsto\Int\lan\mbf,\bphi\ran=\Int f_a\cdot\phi^a.
\end{equation*}

\section{The Lie\,-Poisson structure.} 

\subsection{Definition.}
The {\it Lie-Poisson structure} over the differential algebra $(\cF,D)$
is a bilinear mapping ({\it Lie-Poisson bracket})
\begin{equation*} 
	\{\cdot,\cdot\} : \fF\times\fF\to\fF, \quad \Int K, \Int L\mapsto\big\{\Int K,\Int L\big\}, 
\end{equation*}
with the properties: 
\begin{itemize} 
	\item
		$\big\{\Int K,\Int L\big\}+\big\{\Int L,\Int K\big \}=0$; 
		\hfill ({\it skew-symmetry})
	\item
		$\bfJ\bfI\big(\Int K,\Int L,\Int M\big)
			=\big\{\Int K,\big\{\Int L,\Int M\big\}\big\} +\text{c.p.}=0$; 
			\hfill ({\it Jacobi identity})
\end{itemize}
where the abbreviation ``c.p.'' stands for the cyclic permutation 
of arguments $\Int K,\Int L,\Int M\in\fF$. 
In this case, the pair $(\fF,\{\cdot,\cdot\})$ is a Lie algebra. 
\begin{rem}
The Lie-Poisson bracket is not a Poisson bracket in the proper sense, 
it lacks the {\it product rule} property: 
\begin{equation*}
	\{\Int K,\Int L\cdot\Int M\}=\{\Int K,\Int L\}\cdot\Int M+\Int L\cdot\{\Int K,\Int M\}, 
\end{equation*}
the product  $\Int K\cdot\Int L$ is not even defined in the $\bbF$-linear space $\fF$. 
\end{rem}

\begin{defi}
We define the bracket $\{\cdot,\cdot\} : \fF\times\fF\to\fF$ as follows: 
\begin{equation*} 
	\{\Int K,\Int L\}=\Int\lan\bde K,\bLa\bde L\ran
		=\Int\lan\p K,\bLa\p L\ran, \quad \Int K,\Int L\in\fF,
\end{equation*}
where the $\bbF$-linear mapping 
$\bLa : \oF_\bA\to\cF^\bA$,  $\mbf\mapsto\bphi=\bLa\mbf$, has the properties: 
$\bLa\com\na^*=0$ and $\na\com\bLa=0$. 
\end{defi} 
The definition is correct due to the last two equalities, because in this case 
\begin{equation*} 
	\Int\lan\mbf+\na^*\oF^\M_\bA,\bLa(\bg+\na^*\oF^\M_\bA)\ran
	=\Int\lan\mbf,\bLa\bg\ran \quad\text{for all}\quad \mbf,\bg\in\oF_\bA.
\end{equation*} 

\begin{prop} 
For all $K,L\in\cF$ the representation 
\begin{equation*} 
	\Int\lan\p K,\bLa\p L\ran=\Int\ev_{\bphi(L)}K, \quad \bphi(L)=\bLa\p L\in\bcE, 
\end{equation*}
holds.
\end{prop}

The skew-symmetry and the Jacobi identity impose additional restrictions on the mapping $\bLa$.  

The skew-symmetry reduces to the property 
\begin{equation*}
	\Int(\lan\p K,\bLa\p L\ran+\lan\p L,\bLa\p K\ran)=0 \quad\text{for all}\quad K,L\in\cF,
\end{equation*}
it surely will be performed if the mapping $\bLa$ is the Lagrange skew-adjoint, i.e. if 
$\Int(\lan\mbf,\bLa\bg\ran+\lan\bg,\bLa\mbf\ran)=0$ for all $\mbf,\bg\in\oF_\A$. 
The last condition is also necessary if the image $\im\p=\p\cF\subset\oF_\bA$ 
is of the du Bois-Reymond type in the proper sense.  

The Jacobi identity is much more complicate. 

First let us recall one simple fact from the cohomologies of Lie algebras. 
Let $V$ be a $\bbF$-linear space and 
\begin{equation*}
	\Om(V)=\oplus_{q\in\bbZ_+}\Om^q(V), \quad 
	\Om^0(V)=V, \quad \Om^q(V)=\Hom_\bbF(\w^q_\bbF V;V), \quad q\in\bbN, 
\end{equation*}
be the graded linear space of forms over $V$ with values in $V$. 
Suppose that a skew-symmetric bilinear operation 
$[\cdot,\cdot] : V\times V\to V$ is defined. 
Let us define the graded linear endomorphism 
\begin{equation*}
	d\in\End_\bbF(\Om(V)), \quad d\big|_{\Om^q(V)} : \Om^q(V)\to\Om^{q+1}(V), 
	\quad \om\to d\om, 
\end{equation*}
by the Cartan formula: 
\begin{align*} 
	d\om(u_0,\dots,u_q)
		&=\frac1{q+1}\bigg\{
			\sum_{0\le r\le q}(-1)^r\big[u_r,\om(u_0,\dots\ck{u_r}\dots,u_q)\big] \\
		&+\sum_{0\le r<s\le q}(-1)^{r+s}\om([u_r,u_s],u_0,\dots\ck{u_r}\dots\ck{u_s},\dots,u_q)	
		                      \bigg\},
\end{align*}
where $u_0,\dots,u_q\in V$. 
Then {\it the endomorphism $d$ is an exterior differential, i.e., $d\com d=0$, iff the Jacobi identity 
\begin{equation*} 
	\bfJ\bfI(u,v,w)=[u,[v,w]]+\text{c.p.}=0, \quad u,v,w\in V, 
\end{equation*} 
holds} (i.e., $V$ has the structure of a Lie algebra). 

In particular, let $u\in V=\Om^0(V)$ then 
\begin{equation*}
	du(v)=[v,u] \quad \text{and}\quad ((d\com d)u)(v,w)=\frac12\bfJ\bfI(u,v,w), 
		\quad v,w\in V. 
\end{equation*}

In our case, $V=\fF$ and $[\cdot,\cdot]=\{\cdot,\cdot\}$. 
Thus, the following statement is valid. 
\begin{prop}
The Jacoby identity for the skew-symmetric bilinear mapping 
$\{\cdot,\cdot\} : \fF\times\fF\to\fF$,  $\Int K, \Int L\mapsto\big\{\Int K,\Int L\big\}$ 
holds iff the corresponding linear mapping $d : \Om(\fF)\to\Om(\fF)$ 
is the exterior differential. 
\end{prop}  

To get any practical results concerning the Jacobi identity 
we need additional technical instruments and  assumptions.

\subsection{Linear differential operators.}
For a multi-index $i=(i^1,\dots, i^m)\in\bbI$ we set $D^i=(D_1)^{i^1}\dots(D_m)^{i^m}$. 

We denote by 
\begin{equation*}
	\cF[\D]=\big\{P(D)=P_i\cdot D^i \ \big| \ P=(P_i)\in\oF_\bbI\big\} 
\end{equation*}
the unital associative algebra of all polynomials  ({\it differential polynomials}) 
in the indeterminates $D=(D_1,\dots,D_m)$ with coefficients in $\cF$. 
Every polynomial  $P(D)\in\cF[D]$ defines the linear mapping  
({\it linear differential operator}) 
\begin{equation*}
	P(D) : \cF\to\cF, \quad L\mapsto P(D)L=P_i\cdot D^i L. 
\end{equation*}
{\it The multiplication in $\cF[D]$ is defined by the composition rule 
of the differential operators. }

In the same way, for index sets $\bA,\bB$ we denote by 
\begin{equation*}
	\cF^\bA_\bB[D]=\big\{P(D)=P_i\cdot D^i=(P^a_{ib}\cdot D^i) 
	\ \big| \ P=(P^a_{ib})\in\oF^\bA_{\bbI\bB}\big\} 
\end{equation*}
the $\cF[D]$-module of all polynomials in the indeterminates 
$D=(D_1,\dots,D_m)$ with coefficients in $\cF^\bA_\bB$. 
Every polynomial $P(D)\in\cF^\bA_\bB[D]$ defines the linear mappings: 
\begin{align*}
	&P(D) : \cF^\bB\to\cF^\bA, \quad \bpsi=(\psi^b)\mapsto\bphi=(\phi^a), 
		\quad \phi^a=P^a_{ib}\cdot D^i\psi^b; \\
	&P(D) : \oF_\bA\to\oF_\bB, \quad \mbf=(f_a)\mapsto\bg=(g_b), 
		\quad g_b=P^a_{ib}\cdot D^if_a. 
\end{align*} 
For example, 
\begin{itemize} 
	\item 
		$D\in\cF_\M[D], \quad D : \cF\to\cF_\M, \quad K\mapsto DK, 
		\quad (DK)_\mu=D_\mu K, \quad \mu\in\M$; 
	\item 
		$\Div\in\cF_\M[D], \quad\Div : \cF^\M\to\cF, \quad 
		\psi=(\psi^\mu)\mapsto\Div\psi=D_\mu\psi^\mu$; 
	\item 
		$\na\in\cF^\bA_{\M\bA}, \quad 
		\na : \cF^\bA\to\cF^\bA_\M, \quad \bphi=(\phi^a)\mapsto\bet=(\eta^a_\mu), \\
		 \eta^a_\mu=D_\mu\phi^a+\Ga^a_{\mu b}\cdot\phi^b$; 
	\item 
		$\na^*\in\cF^\bA_{\M\bA}, \quad
		\na^* : \oF^\M_\bA\to\oF_\bA, \quad \bchi=(\chi^\mu_a)\mapsto\mbf=(f_a), \\
		\quad f_a=-D_\mu\chi^\mu_a+\Ga^b_{\mu a}\cdot\chi^\mu_b$. 
\end{itemize} 
Every differential  polynomial 
\begin{equation*}
	P(D)=P_i\cdot D^i : \cF\to\cF, \quad L\mapsto P(D)L=P_i\cdot D^i L,
\end{equation*}
has the {\it Lagrange dual} polynomial 
\begin{equation*} 
	P^*(D)=P^*_i\cdot D^i=(-D)^i\com P_i : \cF\to\cF, \quad 
	K\mapsto P^*(D)K=(-D)^i\big(P_i\cdot K\big), 
\end{equation*}
and the {\it Green's formula} 
\begin{equation*} 
	K\cdot P(D)L-P^*(D)K\cdot L=\Div\psi  
 \end{equation*} 
 holds, where the ``current'' $\psi=(\psi^\mu)\in\cF^\M$, $\psi^\mu=\psi^\mu(P,K,L)$, 
 can be calculated explicitly using the {\it integration by parts} method. 
 In the general case, when the polynomial $P(D)\in\cF^\bA_\bB[D]$ the situation is similar. 
 For example, let 
 \begin{equation*} 
 	P(D)=\big(P^{ab}_i\cdot D_i\big) : \oF_\bA\to\cF^\bA, \quad 
 	\bg=(g_a)\mapsto\bphi=(\phi^a),   
 \end{equation*}
 where $\phi^a=(P(D)\bg)^a=P^{ab}_i\cdot D^i g_b$ (here, $a,b\in\bA$). 
 Then the Lagrange dual polynomial 
 \begin{equation*} 
 	P^*(D)=\big((-D)^i\com P^{ba}_i\big) : \oF_\bA\to\cF^\bA, \quad 
 	\mbf=(f_a)\mapsto\bxi=(\xi^a), 
 \end{equation*}
 where $\xi^a=(P^*(D)\mbf)^a=(-D)^i\big(P^{ba}_i\cdot f_b\big)$, 
 and the Green's formula takes the form 
 \begin{equation*} 
 	\lan\mbf,P(D)\bg\ran-\lan P^*(D)\mbf,\bg\ran=\Div\psi,
 \end{equation*}
 the ``current'' $\psi\in\cF^\M$ is calculated explicitly. 
 
 In particular, $\na^*$ is the Lagrange dual for $\na$, 
 while $\Div$ is the Lagrange dual for $-D$.  
 
 \begin{prop}\label{P1} 
 For every polynomial $P(D)=\big(P^a_{ib}\cdot D^i\big) : \cF^\bB\to\cF^\bA$, 
 its Lagrange dual $P^*(D)=\big((-D)^i\com P^a_{ib}\big) : \oF_\bA\to\oF_\bB$ 
 and any $\bphi\in\bcE$ the following statements hold: 
 \begin{itemize} 
 	\item 
 		$[\ev_{\bphi}, P(D)]=\ev_{\bphi}P(D)=\big((\ev_{\bphi}P^a_{ib})\com D^i\big)$; 
	\item 
		$[\ev_{\bphi},P^*(D)]=\ev_{\bphi}P^*(D)
			=\big((-D)^i\com(\ev_{\bphi}P^a_{ib})\big)=[\ev_{\bphi},P(D)]^*$, 
 \end{itemize} 
 i.e., the evolutionary differentiation $\ev_{\bphi}$ acts coefficient-wise. 
 \end{prop} 
 \begin{proof} 
 The proof based on the characteristic property of the evolutionary differentiations: 
 $[D_\mu,\ev_{\bphi}]=0$ for all $\mu\in\M$ and $\bphi\in\bcE$.
 \end{proof}

\subsection{The additional assumptions.} 
To move further we need the following two additional assumptions. 
\begin{ass} 
We assume that there exists a linear differential operator $\bj=\bj(D) : \cF^\A\to\cF^\bA$, 
$\A$ is an index set, s.t. 
\begin{itemize}
	\item 
		the composition $\na\com\bj=0$, i.e., $\im\bj\subset\ke\na$; 
	\item 
		the commutator $[\ev_{\bphi},\bj]=0$ for any $\bphi\in\bcE$.
\end{itemize}
In more detail, $\bj(D)=(j^a_{i\al}\cdot D^i)\in\cF^\bA_\A[D]$, 
$\phi=(\phi^\al)\mapsto\bj\phi=\bphi=(\phi^a)$, $\phi^a=j^a_{i\al}\cdot D^i\phi^\al$.
\end{ass} 
We define $\cE=\cF^\A$ and $\cE^*=\oF_\A$. 

\begin{prop} 
The following statements hold: 
\begin{itemize} 
	\item 
		the Lagrange dual polynomial  $\bj^*=\big((-D)^i\com j^a_{i\al}\big)\in\cF^\bA_\A[D]$ 
		is defined, $\bj^* : \oF_\bA\to\cE^*$, 
		$\mbf=(f_a)\mapsto f=(f_\al)$, $f_\al=(-D)^i\big(j^a_{i\al}\cdot f_a\big)$; 
	\item 
		the composition $\bj^*\com\na^*=0$, i.e., $\im\na^*\subset\ke\bj^*$; 
	\item 
		the commutator $[\ev_{\bphi},\bj^*]=[\ev_{\bphi},\bj]^*=0$ for any $\bphi\in\bcE$; 
	\item 
		the linear mapping  $\de=\bj^*\com\bde=\bj^*\com\p : \fF\to\cE^*$ 
		acts by the rule: \\
		$\Int L\mapsto\de L=(\de_\al L)$, \quad 
		$\de_\al L=(-D)^i\big(j^a_{i\al}\cdot\p_a L\big)$, $\al\in\A$; 
	\item 
		$\Int\lan\de L,\phi\ran=\Int L_*\phi=\Int\ev_{\bj\phi}L$, $L\in\cF$, $\phi\in\cE$,
		where \\ $L_* : \cE\to\cF$,  
		$\phi=(\phi^\al)\mapsto L_*\phi=\p_a L\cdot j^a_{i\al}\cdot D^i\phi^\al$.
\end{itemize}
\end{prop}

\begin{proof}  
The proof is based on the definitions, Proposition \ref{P1} and the du Bois-Reymond 
property of the algebra $\cF$.
\end{proof}

\begin{ass} 
We assume that the mapping $\bLa : \oF_\bA\to\cF^\bA$ has the form 
$\bLa=\bj\com\La\com\bj^*$, where the skew-adjoint polynomial 
$\La=(\La^{\al\be}_i\cdot D^i)\in\cF^{\A\A}[D]$, while the Lagrange dual polynomial 
$\La^*=((-D)^i\com\La^{\al\be}_i)=-\La\in\cF^{\A\A}[D]$, and 
$\La,\La^* : \cE^*=\oF_\A\to\cE=\cF^\A$.
\end{ass} 

\begin{prop} 
Under the above assumptions $\bLa\in\cF^{\bA\bA}[D]$ 
is the skew-adjoint polynomial, and compositions $\na\com\bLa=0$, $\bLa\com\na^*=0$.
\end{prop} 

\begin{prop}\label{P2}
For every $\bphi\in\bcE$ the equality $[\ev_{\bphi},\La]^*=-[\ev_{\bphi},\La]$ holds. 
\end{prop} 
\begin{proof} 
See Proposition \ref{P1}. 
\end{proof}

\subsection{The Hamiltonian mappings.} 
We {\it assume that all five assumptions listed above hold}. 

In this case the  Lie-Poisson bracket takes the form: 
\begin{equation*} 
	\{\Int K,\Int L\}=\Int\lan\de K,\La\de L\ran  \quad\text{for all}\quad \Int K,\Int L\in\fF.
\end{equation*} 
By the construction this bracket is skew-symmetric, and our aim here to find 
a possibly simple and effective test for the skew-adjoint differential polynomial 
$\La : \cE^*\to\cE$ to be {\it Hamiltonian}, i.e., to satisfy the Jacobi identity. 

For every $R\in\cF$ we denote $\phi(R)=\La\de R\in\cE$ and $\bphi(R)=\bj\phi(R)\in\bcE$. 

\begin{lemma}\label{L1} 
For all $K,L,M\in\cF$ the following equality 
\begin{equation*} 
	\Int\lan\ev_{\bphi(K)}\de L,\phi(M)\ran=\Int\lan\ev_{\bphi(M)}\de L,\phi(K)\ran
\end{equation*}
holds.
\end{lemma}
\begin{proof} 
Indeed, 
\begin{align*} 
	&\Int\lan\ev_{\bphi(K)}\de L,\phi(M)\ran=\Int\lan\ev_{\bphi(K)}(\bj^*\com\p)L,\phi(M)\ran 
		=([\ev_{\bphi(K)},\bj^*]=0)= \\ 
	&=\Int\lan\bj^*(\ev_{\bphi(K)}\p L),\phi(M)\ran=(\, \bj\phi(M)=\bphi(M))
		=\Int\lan\ev_{\bphi(K)}\p L,\bphi(M)\ran= \\ 
	&=\Int\phi^a(K)\cdot(\p_a\p_b L)\cdot\phi^b(M)
		=\Int\phi^b(M)\cdot(\p_b\p_a L)\cdot\phi^a(K)= \\
	&=\Int\lan\ev_{\bphi(M)}\de L,\phi(K)\ran.
\end{align*}
\end{proof}

\begin{theorem}\label{MT} 
The Jacoby identity has the following representation: 
\begin{align*} 
	\bfJ\bfI(\Int K,\Int L,\Int M)
		&=\Int\lan\de K,\La\de\lan\de L,\La\de M\ran\ran+\text{\rm c.p.}	\\
		&=\Int\lan\de K,[\ev_{\bphi(L)},\La]\de M\ran+\text{\rm c.p.}. 
\end{align*} 
\end{theorem} 
\begin{proof} 
Indeed, 
\begin{align*} 
	&\bfJ\bfI(\Int K,\Int L,\Int M)=\Int\lan\de K,\La\de\lan\de L,\La\de M\ran\ran+\text{\rm c.p.}
		=(\La^*=-\La)= \\
	&=-\Int\lan\La\de K,\de\lan\de L,\La\de M\ran\ran+\text{\rm c.p.}
		=(\La\de K=\phi(K),\de=\bj^*\com\p)= \\
	&=-\Int\lan\phi(K),\bj^*\com\p\lan\de L,\La\de M\ran\ran+\text{\rm c.p.}
		=(\bj\phi(K)\cdot\p=\ev_{\bphi(K)})= \\
	&=-\Int\lan\ev_{\phi(K)}\lan\de L,\La\de M\ran\ran+\text{\rm c.p.}
		=(\La\de M=\phi(M))= \\
	&=-\Int\big(\lan\ev_{\bphi(K)}\de L,\phi(M)\ran+\lan\de L,[\ev_{\bphi(K)},\La]\de M\ran
		+\lan\de L, \La\ev_{\bphi(K)}\de M\ran\big)+\text{\rm c.p.}= \\
	&=-\Int\big(\lan\ev_{\bphi(K)}\de L,\phi(M)\ran+\lan\de L,[\ev_{\bphi(K)},\La]\de M\ran
		-\lan\ev_{\bphi(K)}\de M,\phi(L)\ran+ \\
	&+\lan\ev_{\bphi(L)}\de M,\phi(K)\ran+\lan\de M,[\ev_{\bphi(L)},\La]\de K\ran
		-\lan\ev_{\bphi(L)}\de K,\phi(M)\ran+ \\
	&+\lan\ev_{\bphi(M)}\de K,\phi(L)\ran+\lan\de K,[\ev_{\bphi(M)},\La]\de L\ran
		-\lan\ev_{\bphi(M)}\de L,\phi(K)\ran\big)= \\
	&=(\text{Lemma \ref{L1}})=-\Int\lan\de L,[\ev_{\bphi(K)},\La]\de M\ran+\text{\rm c.p.}
		=(\text{Proposition \ref{P2}})= \\
	&=\Int\lan\de M,[\ev_{\bphi(K)},\La]\de L\ran+\text{\rm c.p.} 
		=\Int\lan\de K,[\ev_{\bphi(L)},\La]\de M\ran+\text{\rm c.p.}.
\end{align*}
\end{proof}

\begin{cor} 
If a differential polynomial $\La : \cE^*\to\cE$ is skew-adjoint, 
and the commutator $[\ev_{\bphi},\La]=0$ for all $\bphi\in\bcE$, 
then the polynomial $\La$ is Hamiltonian, i.e., the Jacobi identity holds. 
\end{cor}

\section{The main example -- evolution system without constraints.}
\subsection{The algebra.}
Here (cf., \cite{PJO}, Chapter VII), the algebra $\cF=\cC^\infty_{fin}(\X\bU)$, where 
\begin{itemize} 
	\item 
		$\X=\bbR^\M=\{x=(x^\mu)\mid x^\mu\in\bbR, \ \mu\in\M\}$ 
		is the space of independent variables; 
	\item 
		$\U=\bbR^\A=\{u=(u^\al)\mid u^\al\in\bbR, \ \al\in\A\}$ 
		is the space of dependent variables; 
	\item  
		$\bU=\bbR^\A_\bbI=\{\bu=(u^\al_i)\mid u^\al_i\in\bbR, \ \al\in\A, \ i\in\bbI\}$ 
		is the space of differential variables, $u^\al=u^\al_0$; 
	\item 
		$\cC^\infty_{fin}(\X\bU)$ is the algebra of all smooth functions on 
		the infinite dimensional space $\X\bU=\X\times\bU$, 
		depending on a finite number of the arguments $x^\mu,u^\al_i$. 
\end{itemize} 

\subsection{The differential algebra.}\label{ODA}
The Lie subalgebra $\fD_H$ has the $\cF$-basis $D=\{D_\mu\mid \mu\in\M\}$, 
where the {\it total partial derivatives} 
\begin{equation*} 
	D_\mu=\p_{x^\mu}+u^\al_{i+(\mu)}\p_{u^\al_i}, 
	\quad i+(\mu)=(i^1,\dots,i^\mu+1,\dots,i^m) ,
\end{equation*}
are characterized by the {\it chain rule}  property: 
\begin{itemize}
	\item 
		$\p_{x^\mu}\big(F(x,\bu)\big|_{\bu=\bj\phi(x)}\big)
		=\big(D_\mu F(x,\bu)\big)\big|_{\bu=\bj\phi(x)}, \quad F\in\cF$, 
	\item 
		$\phi=(\phi^\al(x))=\cF^\A$, \quad $\bj\phi=\bphi=(\phi^\al_i(x))\in\cF^\A_\bbI$, 
		\quad $\phi^\al_i(x)=D^i\phi^\al(x)$. 
\end{itemize}

The Lie subalgebra $\fD_V$ has the $\cF$-basis $\{\p_{u^\al_i}\mid \al\in\A, \ i\in\bbI\}$, 
thus, here, $a=\tbinom{\al}{i}$, $\bA=\tbinom{\A}{\bbI}$, 
and the commutator 
\begin{equation*}
	[D_\mu,\p_{u^\al_i}]=-\p_{u^\al_{i-(\mu)}}, 
	\quad\text{i.e.,}\quad \Ga^{i\be}_{\mu\al j}=-\de^\be_\al\de^i_{j+(\mu)}, 
	\quad \al,\be\in\A, \ i,j\in\bbI, \ \mu\in\M.
\end{equation*} 

\subsection{The key ingredients.}
There are defined the mappings: 
\begin{itemize} 
	\item 
		$\na : \cF^\A_\bbI\to\cF^\A_{\M\bbI}, \quad  
		\bphi=(\phi^\al_i)\mapsto\bet=(\eta^\al_{\mu i}), \quad 
		\eta^\al_{\mu i}=D_\mu\phi^\al_i-\phi^\al_{i+(\mu)}$; 
	\item 
		$\bj : \cF^\A\to\cF^\A_\bbI, \quad \phi=(\phi^\al)\mapsto\bphi=(\phi^\al_i), 
		\quad \phi^\al_i=D^i\phi^\al=\de_{ik}\de^\al_\be\cdot D^k\phi^\be$.
\end{itemize}
One can verify by the induction that 
\begin{equation*} 
	\bcE=\ke\na=\big\{\bphi\in\cF^\A_\bbI \ \big| \ \bphi=\bj\phi, \ \phi\in\cF^\A\big\}
	\simeq\cF^A, 
\end{equation*} 
Let us set $\cE=\cF^\A$, then we get the exact sequence 
\begin{equation*}
	0\rightarrow\cE\xrightarrow{\bj}\cF^\A_\bbI\xrightarrow{\na}\cF^\A_{\M\bbI}, 
	\quad\text{i.e.}\quad \ke\bj=0, \quad \im\bj=\ke\na.
\end{equation*} 
Moreover, the commutator $[\ev_{\bphi},\bj]=0$ for all $\bphi\in\bcE$, 
because the differential polynomial 
$\bj=\bj(D)\in\cF^\A_{\bbI\A}[D]$ has constant coefficients 
$j^\al_{ik\be}=\de_{ik}\de^\al_\be$.
 
Further, the $\cF$-module $\Om^1_H$ has the dual $\cF$-basis $\{dx^\mu\mid \mu\in\M\}$,
and for the $\cF$-module $\Om^1_V$ it is convenient to chose the dual $\cF$-basis 
\begin{equation*}
	\big\{\de u^\al_i=du^\al_i-u^\al_{i+(\mu)}dx^\mu \ \big| \ \al\in\A, \ i\in\bbI\big\}. 
\end{equation*}
In particular, $d_V\de u^\al_i=0$, $d_H\de u^\al_i=-\de u^\al_{i+(\mu)}\w dx^\mu$, 
$\al\in\A$ , $i\in\bbI$.

There are defined the Lagrange dual mappings: 
\begin{itemize} 
	\item
		$\na^* : \oF^{\M\bbI}_\A\to\oF^\bbI_\A, \quad 
	\bchi=(\chi^{\mu i}_\al)\mapsto \mbf=(f^i_\al), \quad 
	f^i_\al=D_\mu\chi^{\mu i}_\al+\de^i_{j+(\mu)}\chi^{\mu j}_\al$;
	\item 
		$\bj^* : \oF^\bbI_\A\to\cF_\A, \quad \mbf=(f^i_\al)\mapsto \bj^*\mbf=f=(f_\al), 
		\quad f_\al=(-D)^i\de_{ik}f^k_\al\quad$.
\end{itemize} 
Let us set $\cE^*=\cF_\A$.

\begin{prop} 
Let $\mbf=(f^i_\al),\bg=(g^i_\al)\in\oF^\bbI_\A$, $g^i_\al=\de^i_0f_\al$, 
$f_\al=(\bj^*\mbf)_\al\in\cE$, 
then 
\begin{equation*} 
	\mbf-\bg=\na^*\bchi, \quad\text{where}\quad \bchi=(\chi^{\mu i}_\al)\in\oF^{\M\bbI}_\A, 
	\quad \chi^{\mu i}_\al=-\frac1m\de_{kj}(-D)^k f^{i+j+(\mu)}_\al,  
\end{equation*}
the summation is over all $k,j\in\bbI$,  
remind that only a finite number of components $f^j_\al\ne0$.
\end{prop}
\begin{proof} 
The proof is done by the direct check.
\end{proof}

\begin{cor} 
The linear space $\oF^\bbI_\A$ splits into the direct sum of the linear spaces, 
$\oF^\bbI_\A=\cE^*\oplus_\bbF\na^*\oF^{\M\bbI}_\A$, 
where the injection mapping $\cE^*\to\oF^\bbI_\A$ acts by the rule: 
$f=(f_\al)\mapsto\mbf=(f^i_\al)$, $f^i_\al=\de^i_0f_\al$.
\end{cor}

\begin{cor}
The following sequence of the linear spaces 
\begin{equation*}
	0\leftarrow\cE^*\xleftarrow{\bj^*}\oF^\bbI_\A=\cE^*\oplus\na^*\oF^{\M\bbI}_\A
		\xleftarrow{\na^*}\oF^{\M\bbI}_\A
\end{equation*} 
is exact. 
\end{cor}

\begin{cor} 
There is the representation $\cS=\oF^\bbI_\A\big/\na^*\oF^{\bbI\M}_\A=\cE^*$.
\end{cor} 

\subsection{The verification of the assumptions.} 
Let us check that here all five assumptions are fulfilled. 
\begin{itemize} 
	\item 
		The Lie algebra $\fD=\fD(\cF)$ is splitted 
		into the direct sum of vertical and horizontal subalgebras, 
		$\fD=\fD_V\oplus_\cF\fD_H$. 
	\item 
		For every $f\in\cF$ the derivative $\p_{u^\al_i}f\ne0$ only for a finite number 
		of indices  $a=\tbinom{\al}{i}\in\bA=\tbinom{\A}{\bbI}$. 
	\item 
		The differential algebra $(\cF,D)$, $\cF=\cC^\infty_{fin}(\X\bU)$, 
		is of the du Bois-Reymond type (see, e.g., \cite{PJO}). 
	\item 
		The mapping $\bj : \cE\to\bcE$ here enjoys the properties: 
		$\im\bj=\ke\na$ and $[\ev_{\bphi},\bj]=0$ for all $\bphi\in\bcE$, 
		so Assumption 4 is fulfilled. Moreover, here $\im\na^*=\ke\bj^*$, also. 
	\item 
		The Assumption 5, in fact, is a choice of the representation for the mapping $\bLa$, 
		it doesn't lead to any additional restrictions.
\end{itemize}

\section{Conclusion.} 

Let us summarize. 
\begin{itemize} 
	\item 
		It was shown that the Lie-Poisson structures admit a pure algebraic formulation 
		as Hamiltonian mappings over differential algebras. 
	\item 
		The main example, in particular, proves that the five assumptions, listed above, 
		are compatible and up to date. 
	\item 
		The crucial component of the scheme is the symbol 
		$\Ga=(\Ga^b_{\mu a})\in\cF^\bA_{\M\bA}$, 
		determining all further constructions. 
	\item 
		The main (and the most difficult) problem at this stage is to find 
		a mapping $\bj : \cF^\A\to\cF^\bA$ with the necessary properties 
		(i.e., a {\it jet mapping}). 
\end{itemize} 
What are possible applications of our construction? 
Suppose that an evolution system of partial differential equations with constraints 
is given and one wants to try to represent it as a Hamiltonian system.  
To apply the above scheme in this situation one should to specify an appropriate 
differential algebra compatible with constraints and satisfying to the three basic assumptions. 
The natural way to do this provides the algebraic approach to partial differential equations 
(see, e.g., \cite{Z1}). After that one should act as it is described in the introduction. 
Note, that applications of our construction are not restricted to evolution systems 
of partial differential equations. In particular, this method can be applied, 
e.g., in the following researches: \cite{Z3}, \cite{Z4}, \cite{AS1}, \cite{AS2}, 
\cite{MK1}, \cite{MK2}, \cite{MS},  \cite{AKG},\cite{DVB}, \cite{AKP}, \cite{AAS}, etc..

\end{document}